\documentclass[11pt]{article}
\usepackage{amsmath,amssymb,fullpage,graphicx}

\newtheorem{definition}{Definition}

\newtheorem{theorem}{Theorem}
\newtheorem{lemma}[theorem]{Lemma}

\newtheorem{claim}[theorem]{Claim}

\newtheorem{conjecture}[theorem]{Conjecture}

\newenvironment{proof}{\noindent{\bf Proof:} \hspace*{1mm}}{
        \hspace*{\fill} $\Box$ }
\newenvironment{proof_of}[1]{\noindent {\bf Proof of #1:}
        \hspace*{1mm}}{\hspace*{\fill} $\Box$ }

\title{On Learning Finite-State Quantum Sources}
\author{Brendan Juba\footnote{%
Supported by a NSF Graduate Research Fellowship.}
\\ MIT CSAIL\\ {\tt bjuba@mit.edu}}

\newcommand{\toffoli}{\wedge_\oplus}
\newcommand{\ket}[1]{|#1\rangle }
\newcommand{\poly}{\mathrm{poly}}
\newcommand{\tP}{\tilde{P}}
\newcommand{\tx}{\tilde{x}}
\newcommand{\bE}{\mathbb{E}}
\newcommand{\calF}{\mathcal{F}}
\newcommand{\calC}{\mathcal{C}}
\newcommand{\calD}{\mathcal{D}}
\newcommand{\calP}{\mathcal{P}}
\newcommand{\bbC}{\mathbb{C}}

\begin{document}
\maketitle

\begin{abstract}
We examine the complexity of learning the distributions produced by finite-state
quantum sources. We show how prior techniques for learning hidden Markov models
can be adapted to the {\em quantum generator} model to find that the analogous
state of affairs holds: information-theoretically, a polynomial number of
samples suffice to approximately identify the distribution, but computationally,
the problem is as hard as learning parities with noise, a notorious open
question in computational learning theory.
\end{abstract}

\section{Introduction}
In recent work, Wiesner and Crutchfield~\cite{wc08} introduced {\em Quantum
Generators} as a formal model of simple quantum mechanical systems. In this
model, a simple quantum mechanical system is observed repeatedly, yielding a
classical stochastic process consisting of the sequence of discrete measurement
outcomes, analogous to how an underlying Markov process yields a sequence of
observations in a hidden Markov model. From this perspective, it is natural to
wonder what can be learned about such a simple quantum mechanical system from
the sequence of measurement outcomes.

In this work, we consider the question of whether or not it is feasible to learn
the distribution on measurement outcomes from a reasonable (polynomially 
bounded) number of observations. We state two theorems on this subject: first,
in Section~\ref{PAC-learn}, we show that it is information-theoretically 
possible to learn the distribution over measurements for binary processes
in polynomially many observations, but we then show in Section~\ref{hardlearn}
that under a standard hardness assumption (Conjecture~\ref{noisyparity}, that it
is computationally infeasible to learn parity functions in the presence of
classification noise) that it is also computationally infeasible to learn the
output distribution of a Quantum Generator (also for a binary alphabet).

\section{Preliminaries}
We begin by recalling the formal definition of Quantum Generators (specialized
to binary observations here) and the models of learning that we will need.

\subsection{The Quantum Generator Model}

Quantum Generators, defined by Wiesner and Crutchfield~\cite{wc08}, are a model
of a simple, repeatedly observed quantum mechanical system. Formally:
\begin{definition}[Quantum Generator]
A {\em $k$-state Quantum Generator} is given by a four-tuple, 
$(\ket{\psi_0},U,M,\Sigma)$ where the {\em initial state} 
$\ket{\psi_0}\in\bbC^k$ has $\ell_2$-norm 1, $U$ is a unitary transformation on
$\bbC^k$, $\Sigma$ is a finite set of {\em measurement outcomes}, and $M$ is a 
{\em projective measurement operator}, i.e., there is a partition of 
$\{1,\ldots,k\}$ into $|\Sigma|$ sets such that associated with each $\sigma\in
\Sigma$, there is a projection $M_\sigma$ onto the associated coordinates.

A Quantum Generator produces a probability distribution in the following way:
given $\ket{\psi_t}$, for each $\sigma\in\Sigma$, $x_{t+1}=\sigma$ and 
$\ket{\psi_{t+1}}=\frac{M_\sigma U\ket{\psi_t}}{||M_\sigma U\ket{\psi_t}||_2}$
with probability $||M_\sigma U\ket{\psi_t}||_2^2$. Thus, in particular, the
probability of the $n$-symbol output $x_1,\ldots,x_n\in\Sigma^n$ is given by
$||M_{x_n} U\cdots M_{x_1} U\ket{\psi_0}||_2^2$.
\end{definition}
In this work, {\em we will only consider measurements with two output symbols}.
Thus, in general (if the system has more than two basis states), we only
consider degenerate measurements. This is, of course, with some loss in 
generality, but it also means that the hardness result in Theorem~\ref
{comp-hard} holds even for a highly restricted class.

From a theoretical perspective, it is also natural to wonder if it is necessary
to link the output distribution and measurement of the quantum system -- and
certainly, proposals for formal models that do not identify these two concepts
exist in the literature~\cite{fw01,gudder99} -- but in their work, Wiesner and
Crutchfield stress that the resulting (alternative) models do not capture simple
physical systems. Since we wish to strive for relevance in this case, we adopt
the model of Wiesner and Crutchfield here. Again, we also stress that our
negative result holds even for this more restricted class of (physically 
relevant) processes.

We also remark that we allow our Quantum Generators to start in an arbitrary
state and in the model of learning distributions that we consider, we assume
that it is possible to take many independent samples from this distribution.
This is arguably unrealistic, but we note that the hardness result is likely to
be more relevant to practice, where the construction we use in our
hardness result turns out to have two desirable properties: first, it starts in
a basis state (i.e., of the form $e_i$), and second, the $mn$-symbol 
distribution of the Quantum Generator is distributed identically to $m$
independent copies of the $n$ symbol distribution, so we also have hardness for
learning from a single, long sample as well. For more details, consult
Appendix~\ref{hardness-proof}.
\subsection{Models of learning distributions}
In contrast to the classic PAC model, and in contrast to the approach taken
by Abe and Warmuth in their treatment of probabilistic automata~\cite{aw92},
our positive and negative results will all be given for the
representation-independent ``improper PAC'' distribution-learning model
introduced by Kearns et al.~\cite{kmrrss94}. Specifically, we use their
notion of learning with an evaluator:
\begin{definition}[Distribution learning under the KL-divergence]
We say that a class of distributions $\calD$ is {\em learnable under the
KL-divergence} in $m$ samples (time complexity $t$) if there is an algorithm 
that, on input $n$, $\varepsilon$, $\delta$, and $x_1,\ldots,x_m\in\{0,1\}^n$ 
sampled from $D_n$ for $D=\{D_n\}_n$ an ensemble from $\calD$, outputs an
``evaluator'' circuit $E:\{0,1\}^n\rightarrow [0,1]$ (within $t$ steps) such 
that the distribution on $\{0,1\}^n$ computed by $E$ satisfies $KL(D_n||E)<
\varepsilon$ with probability $1-\delta$.
\end{definition}
We will comment explicitly on the time efficiency of the learning algorithm and
number of samples $m$, as appropriate. In particular, if $m$ is an appropriate
polynomial (in $n$, $\frac{1}{\varepsilon}$, $\log\frac{1}{\delta}$, and in our
case also $k$, the number of states), this corresponds to improper PAC-learning,
and if $t$ is an appropriate polynomial (in the same parameters) then learning
is said to be {\em efficient}.

We also use a hardness of learning assumption, which depends on
the definition of learning in the presence of noise~\cite{al88}:
\begin{definition}[Learning in the presence of noise]
We say that a class of boolean functions $\calC$ is {\em efficiently learnable 
under the uniform distribution with noise rate $\eta$} if there is an algorithm 
that, on input $n$, $\varepsilon$, $\delta$, and $\eta$, when given $x_1,\ldots,
x_m$ uniformly chosen from $\{0,1\}^n$ and $b_1,\ldots,b_m$ where each 
$b_i=f(x_i)$ for a fixed $f\in\calC$, with probability $1-\eta$ independently,
with probability $1-\delta$ outputs the representation of a function 
$f'$ such that $\Pr_{x\in\{0,1\}^n}[f(x)\neq f'(x)]<\varepsilon$, 
in time polynomial in $n,$ $\frac{1}{\varepsilon}$, and $\log\frac{1}{\delta}$.
\end{definition}

\section{Improper PAC-learnability}\label{PAC-learn}
In this section, we adapt the approach used by Abe and Warmuth~\cite{aw92}
to show that (classical) probabilistic automata are PAC-learnable to show that
the distributions produced by Quantum Generators are improperly PAC-learnable
under the KL-divergence. 

Following Kitaev, we employ the set of gates
$\{I,S,K,\bigoplus,\toffoli\}$ where $I$ is the identity gate,
$S=\frac{1+i}{2}\left(\begin{array}{cc}1&1\\1&-1\end{array}\right)$
is a scaled Hadamard gate,
$K=\left(\begin{array}{cc}1&0\\0&i\end{array}\right)$ is a phase shift,
$\bigoplus(\ket{a,b})=\ket{a,a\oplus b}$,
and $\toffoli(\ket{a,b,c})=\ket{a,b,(a\wedge b)\oplus c}$ is a Toffoli gate.
We first recall the Solovay-Kitaev Theorem~\cite{kitaev}
\begin{theorem}[Solovay-Kitaev]\label{skt}
For any $\delta>0$ and $n$-qubit unitary $U$, there is a
${O(2^{2n}(n+\poly\log\frac{1}{\delta}))}$ gate 
$\ell_2$ $\delta$-approximation to $U$ in our set of gates.
\end{theorem}
In particular, since a $k$-state quantum generator has a unitary with a
$\log k$-qubit representation, we find:
\begin{claim}\label{goodnet}
There is an $\epsilon$-net under the $\ell_\infty$ distance 
on the $n$-symbol output distributions of $k$-state Quantum Generators of size 
$2^{\poly(k, n, \log\frac{1}{\epsilon})}$
\end{claim}

The key of Abe and Warmuth's analysis was that for any distributions $P$ and 
$Q$, the KL-divergences of the empirical distributions $\hat{P}_n$ from $Q_n$, 
$KL(\hat{P}_n||Q_n)$ converge to $KL(P_n||Q_n)$ (essentially by Hoeffding's
inequality) where we can calculate the former quantity for a given distribution
$Q$ from our $\epsilon$-net. At this point, the learning algorithm is
essentially obvious; the only problem is that the KL-divergence is infinite for 
strings outside the support of a distribution from our $\epsilon$-net, which
would prevent the use of the concentration result. We avoid this by perturbing
the distributions slightly: in the distribution over $n$-symbol samples, we 
fix the minimum probability that any symbol is output on any step to (roughly)
$\varepsilon/n$ (altering the remaining probabilities accordingly).
It is easy to see that this guarantees an upper bound on the KL-divergence
(between our modified distribution and {\em any} distribution over $n$ symbol
strings) of $n\log\frac{n}{\varepsilon}$. Taking (again, roughly) 
$\epsilon=(\varepsilon/2n)^{2n}$, we can show that for the distribution $\tilde
{D}$ we obtain from our perturbed approximation to a distribution $D$ obtained
from a Quantum Generator, the total KL-divergence from $D$ is at most
$\varepsilon$. Note that the elements of the $\epsilon$-net still have
representations of size polynomial in $n$ since the dependence on $\epsilon$ was
only polylogarithmic. Thus, we find:

\begin{theorem}\label{improper-pac}
The class of $k$-state Quantum Generators is learnable under the KL-divergence 
with sample complexity $\poly(n,k,\frac{1}{\varepsilon},\log\frac{1}{\delta})$.
\end{theorem}
The full proof is given in Appendix~\ref{PAC-proof}.

\section{Computational hardness of learning}\label{hardlearn}
We now show the computational hardness of learning the output distributions
of Quantum Generators, under the assumption that learning noisy
parity functions is hard. More specifically, we say that a function
$f:\{0,1\}^n\rightarrow\{0,1\}$ is a {\em parity function} if there is some 
$S\subset\{1,\ldots,n\}$ such that $f(x)=\bigoplus_{i\in S}x_i$, and we
assume that it is hard to identify the set $S$ when we are given random
examples of $f$ with $f(x)$ negated with some probability $\eta$. 
Formally, the assumption is:

\begin{conjecture}[Noisy Parity Learning]\label{noisyparity}
There is a constant $\eta\in (0,1/2)$ such that no algorithm learns the class
of parity functions with noise rate $\eta$ under the uniform distribution in
time polynomial in $n$, $\frac{1}{\epsilon}$, and $\frac{1}{\delta}$.
\end{conjecture}

These functions are known not to be learnable in the restricted {\em statistical
query model}~\cite{kearns98, bfjkmr94}, which captures most known algorithms for
efficient learning in the presence of classification noise, although the best
known algorithm for the problem, due to Blum, Kalai and Wasserman~\cite{bkw03}
efficiently learns parities up to size $O(\log n\log\log n)$, which is beyond
what can be learned in the statistical query model. (For parities of $\Theta(n)$
bits, however, the algorithm requires $2^{\Omega(n/\log n)}$ samples.) 
Feldman et al.~\cite{fgkp06} recently showed that many other problems not known 
to be learnable in the presence of classification noise reduce to the problem of
learning noisy parities, establishing its central place in the classification
noise model. Moreover, this problem is related to the long-standing open problem
of decoding random linear codes~\cite{bfkl93}, and worse still, Feldman et al.
show that learning parities with random noise is as hard as learning parities in
the agnostic learning (adversarial noise) model~\cite{KSS94}. Thus, in any case,
it represents a serious barrier to the current state of the art, and any
algorithm for our problems of interest would represent a major breakthrough on
numerous fronts.

The result proceeds, simply enough, by showing that a Quantum Generator of
modest size (linear in $n$) can produce exactly the distribution of labeled
examples of a parity function with $\eta$ noise, where learning the distribution
of the parity function is sufficient to learn the parity. The construction is a
modification of the analogous constructions for probabilistic automata and
hidden Markov models given by Kearns et al. and Mossel and Roch, 
respectively~\cite{kmrrss94, mr06}. Our construction is illustrated in 
Figure~\ref{fig-qg-parity}. The result is:

\begin{figure}
\centering
\includegraphics[width=5in]{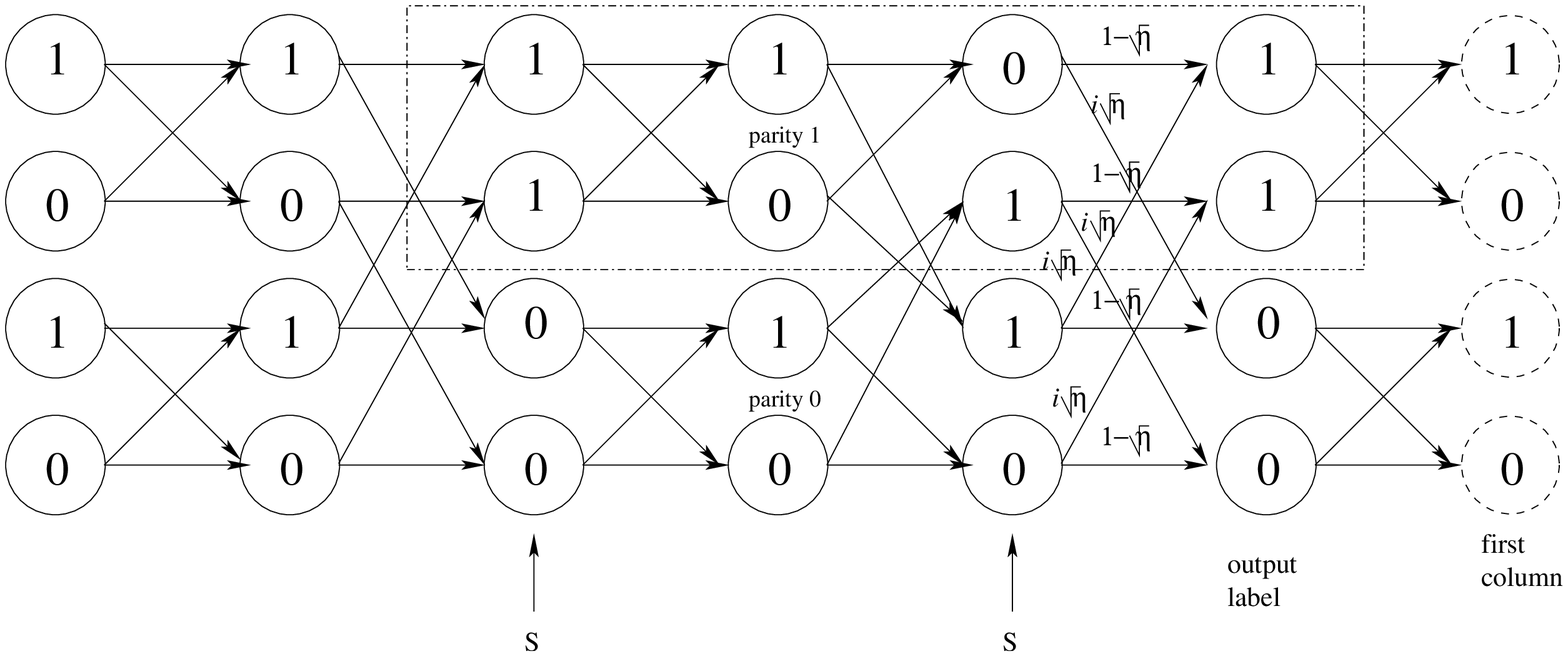}
\caption{A $4(n+1)$-state QG generating a noisy parity of $S=\{3,5\}$ for $n=5$.
Circles correspond to states with labels indicating which partition they belong 
to under the measurement operator; unlabeled transitions come in pairs with 
weights $1/\sqrt{2}$ and $i/\sqrt{2}$.}
\label{fig-qg-parity}
\end{figure}

\begin{theorem}\label{comp-hard}
Assuming the Noisy Parity Learning Conjecture, no algorithm can learn the 
$n$-bit output distribution of a $k$-state Quantum Generator under the 
KL-divergence in time polynomial in $n$, $k$, $\frac{1}{\varepsilon}$, and
$\log\frac{1}{\delta}$.
\end{theorem}
The proof is given in Appendix~\ref{hardness-proof}.

\section*{Acknowledgements}
The author would like to thank Seth Lloyd, Madhu Sudan, and Eran Tromer for
discussions that motivated the questions considered here, and Elad Verbin
for suggesting the relevance of learning noisy parities. The author also thanks 
Scott Aaronson for a smashing course on Quantum Complexity Theory, 
where this work was originally submitted as a course project.

\bibliographystyle{plain}

\bibliography{qct}

\appendix

\section{Proof of improper PAC-learnability}\label{PAC-proof}

For convenience, for a distribution $P$ on $\{0,1\}^n$ 
and sample $x\in\{0,1\}^n$, we define 
$P_i(x)=P(x_i|x_1,\ldots,x_{i-1})$. Thus, $P(x)=\prod_iP_i(x)$.

\begin{proof_of}{Claim~\ref{goodnet}}
Fix a measurement operator $M$ on a quantum system with $k$ basis states,
and consider the Quantum Generator with a unitary $U$ and starting state
$\ket{\psi_0}$. Consider the $\poly(k,\log\frac{1}{\epsilon_0})$-gate
approximation to $U$, $U'$, given by the Solovay-Kitaev Theorem, and
a $2k\log\frac{k}{\epsilon_0}$-bit approximation $\ket{\psi_0'}$ to 
$\ket{\psi_0}$ with representation $(b_1,\ldots,b_k)$ corresponding to
the normalization of the vector 
\[
\left(\left(1-\frac{\epsilon_0}{k}\right)^{b_1},\ldots,
\left(1-\frac{\epsilon_0}{k}\right)^{b_k}\right)
\]
noting that $\left(1-\frac{\epsilon_0}{k}\right)^{\frac{k}{\epsilon_0}
\log\frac{k}{\epsilon_0}}\leq\frac{\epsilon_0}{k}$. We therefore see that 
$\ket{\psi_0}$ has an approximation $\ket{\psi_0'}$ such that each entry is 
within a multiplicative $(1-\frac{\epsilon_0}{k})$-factor unless it is smaller
than $\frac{\epsilon_0}{k}$, so that in either case, the
$\ell_2$ distance between $\ket{\psi_0}$ and $\ket{\psi_0'}$
(recalling that $\ket{\psi_0}$ has $\ell_2$ norm 1) is at most $2\epsilon_0$.
Noting that at each step, the probability of $x_1,\ldots,x_i$ is equal
to the $\ell_2^2$ norm of $M_{x_i}U\cdots M_{x_1}U\ket{\psi}$,
it is easy to see that each application of $U'$ now grows the gap between
$P(x)$ and $P'(x)$ by at most $\epsilon_0$, so the total gap between
$P(x)$ and $P'(x)$ is at most $(n+2)\epsilon_0$.
Since $M$ has a $k$-bit representation and $U'$ has a
$\poly(n,k,\log\frac{1}{\epsilon})$-bit representation,
clearly the overall size of the $\epsilon$-net (taking $\epsilon_0=
\frac{\epsilon}{n+2}$) is $2^{\poly(n,k,\log\frac{1}{\epsilon})}$,
as claimed.
\end{proof_of}

For a fixed $\epsilon_1$, given a distribution $P$ and observation $x$, 
we define the perturbed distribution $\tP(x)$ 
(and associated ``corrected'' observation $\tx$)
as follows: if $P(x_1)<\epsilon_1$, then 
$\tP(x_1)=\epsilon_1$ and similarly, $\tP(x_1)=1-\epsilon_1$ whenever
$P(x_1)>1-\epsilon_1$; if $P(x_1)=0$, then
$\tx_1=\neg x_1$, otherwise, we put $\tx_1=x_1$.
If, on the other hand, $1-\epsilon_1\geq P(x_1)\geq \epsilon_1$, 
$\tP(x_1)=P(x_1)$.
Now, assuming that we have defined $\tx_1,\ldots,\tx_{i-1}$ and 
$\tP(x_1),\ldots,\tP(x_{i-1})$, we similarly define $\tx_i$ to be
$x_i$ if $P(x_i|\tx_1,\ldots,\tx_{i-1})\neq 0$ and $\neg x_i$
otherwise; finally, as before, we put $\tP_i(x)$ equal to
$P(x_i|\tx_1,\ldots,\tx_{i-1})$ ``restricted'' to the range
$[\epsilon_1,1-\epsilon_1]$.

It is easy to see that $\tP$ is a probability distribution over $\{0,1\}^n$.
Moreover, suppose $P'$ is a distribution such that 
$|P'(x)-P(x)|<\epsilon_2$ for all $x$ 
(e.g., as obtained via Claim~\ref{goodnet}).
We then have that $\tP'\geq\epsilon_1^n$ and 
$\tP'_i(x)<P'_i(x)$ only when $\tP'_i(x)=1-\epsilon_1$, and thus
\begin{align*}
KL(P||\tP')&=\sum_xP(x)\sum_i\log\frac{P_i(x)}{\tP'_i(x)}\\
&\leq\sum_{x:P(x)>\epsilon_1^n}P(x)\left[\sum_i\log\frac{P_i(x)}{P'_i(x)}+
\sum_{i:1-\epsilon_1\leq P'_i(x)}\log\frac{1}{1-\epsilon_1}\right]\\
&\leq\sum_{x:P(x)>\epsilon_1^n}P(x)\log\frac{P(x)}{P(x)-\epsilon_2}+
n\log\frac{1}{1-\epsilon_1}\\
&\leq\log\left(1+\frac{\epsilon_2}{\epsilon_1^n-\epsilon_2}\right)
+n\log\left(1+\frac{\epsilon_1}{1-\epsilon_1}\right)\\
&\leq \frac{\epsilon_2}{\epsilon_1^n-\epsilon_2}+n\epsilon_1
\end{align*}
so if we take $\epsilon_1^n-\epsilon_2=\sqrt{\epsilon_2}$, $KL(P||\tP')\leq
\sqrt{\epsilon_2}+n\epsilon_2^{1/2n}(1+\sqrt{\epsilon_2})^{1/n}$. Thus,
for a desired $\varepsilon_0$, taking $\epsilon_2=(\varepsilon_0/2(n+1))^{2n}$
suffices to give $KL(P||\tP')<\varepsilon_0$. Moreover, the size of the
$\epsilon_2$-net is still $2^{\poly(n,k,\log\frac{1}{\varepsilon_0})}$ (with a
larger dependence on $n$) and since $\tP'>\epsilon_1^n$, for every distibution 
$Q$ over $\{0,1\}^n$, we find
\[
KL(Q||\tP')=\sum_xQ(x)\log\frac{1}{\tP'}-H(Q)\leq
\sum_xQ(x)\log\frac{1}{\epsilon_1^n}=n\log\frac{1}{\epsilon_1}\leq n\log\frac
{2(n+1)}{\varepsilon_0}
\]

We now recall the following standard lemma used by Abe and Warmuth~\cite{aw92},
following from Hoeffding's inequality. (They reference Pollard~\cite{pollard}.)
\begin{lemma}Let $\calF$ be a finite set of random variables with range bounded
by $[0,M]$. Let $D$ be an arbitrary distribution. Then, if
\[
m\geq\frac{M^2}{\varepsilon^2}(\ln|\calF|+\ln\frac{1}{\delta})
\]
we have
\[
\Pr_{x_1,\ldots,x_m\in D}\left[\exists f\in\calF:\left|\frac{1}{m}\sum_if(x_i)-
\bE_D[f]\right|>\varepsilon\right]<\delta
\]
\end{lemma}
Naturally, if $\calP$ is the set of perturbed distributinos from our $\epsilon_2
$-net, we apply this lemma with $\calF=\{\log\frac{1}{\tP'}:\tP'\in\calP\}$.
Thus, $\ln|\calF|=\poly(n,k,\log\frac{1}{\varepsilon_0})$ and 
$M=n\log\frac{2(n+1)}{\varepsilon_0}$. We also use $\varepsilon_0$ as 
$\varepsilon$, for convenience.

For the corresponding polynomial number of samples we find, following Abe and 
Warmuth, that for the true distribution $P$, its perturbed estimate $\tP'$, and
any perturbed distribution $P^*$ acheiving the minimum value of $\frac{1}{m}
\sum_i\log\frac{1}{P^*(x_i)}$, with probability $1-\delta$, the following
simultaneously hold:
\begin{align*}
\bE_P[\log\frac{1}{P^*}]-\frac{1}{m}\sum_i\log\frac{1}{P^*(x_i)} 
&< \varepsilon_0\\
\frac{1}{m}\sum_i\log\frac{1}{\tP'(x_i)}-\bE_P[\log\frac{1}{\tP'}] 
&< \varepsilon_0\\
\frac{1}{m}\sum_i\log\frac{1}{P^*(x_i)}-\frac{1}{m}\sum_i\log\frac{1}{\tP'(x_i)}
&\leq 0
\end{align*}
by summing the three, we find
\[
\bE_P[\log\frac{1}{P^*}]-\bE_P[\log\frac{1}{\tP'}]<2\varepsilon_0
\]
so therefore $KL(P||P^*)-KL(P||\tP')<2\varepsilon_0$.
Since we argued above that $KL(P||\tP')<\varepsilon_0$, we find that
$KL(P||P^*)<3\varepsilon_0$, so by taking $\varepsilon_0$ sufficiently small,
we see that it is sufficient to output a circuit corresponding to this
$P^*$. Since evaluating $P^*$ from its gate construction merely involves
performing a polynomial number of matrix operations to polynomial precision, 
Theorem~\ref{improper-pac} follows.

\section{Proof of computational hardness}\label{hardness-proof}
Let any parity function $f_S$ and any noise rate $\eta\in (0,1/2)$ be given.
Following the constructions of Kearns et al.~\cite{kmrrss94} and Mossel and
Roch~\cite{mr06}, we describe a $4(n+1)$-state Quantum generator for which the 
$(n+1)$-symbol output distribution is precisely the noisy parity 
distribution---$(x,f_S(x)\oplus b)$ where $x\in\{0,1\}^n$ is uniformly chosen
and $b\in\{0,1\}$ has $b=1$ with probability $\eta$.

\paragraph{Construction:} For convenience, we will index the basis states by 
$(j,k,\ell)\in\{0,1,\ldots,n\}\times\{0,1\}\times\{0,1\}$, where (cf. 
Figure~\ref{fig-qg-parity}) we think of $j$ as representing a column, $k=1$ as
representing the ``top half,'' and $\ell=1$ as representing the ``upper state.''
We will explicitly describe the entries of the matrix representation of the
Quantum Generator's unitary. (Verifying next that the matrix actually describes
a unitary transformation, of course!)

For each column $(j,k,\ell)$, there are exactly two nonzero entries, each in 
rows of the form $(j+1\bmod{n+1},k',\ell')$. For $j=0,\ldots,n-1$, if $(j+1)
\notin S$, then the nonzero entries are $1/\sqrt{2}$ in $(j+1,k,\ell)$ and 
$i/\sqrt{2}$ in $(j+1,k,\ell\oplus 1)$; if $(j+1)=\min(S)$, then the nonzero 
entries are $1/\sqrt{2}$ in $(j+1,k,\ell)$ and $i/\sqrt{2}$ in $(j+1,k\oplus 1,
\ell)$; and, if $(j+1)\in S$ but it is not the minimum element, then the entries
are $1/\sqrt{2}$ in $(j+1,k\oplus\ell,k)$ and $i/\sqrt{2}$ in $(j+1,1\oplus k
\oplus\ell,k)$. Finally, if $j=n$, then the nonzero entries are $\sqrt{1-\eta}$
in $(0,k,\ell)$ and $i\sqrt{\eta}$ in $(0,k\oplus 1,\ell)$. We further observe
that each row also has exactly two nonzero entries, one in column $(j,k,\ell)$
with zero complex part and one in column $(j,k',\ell')$ with zero real part; 
moreover, these two columns appear together in the support of another row, with 
column $(j,k,\ell)$ having zero real part and $(j,k',\ell')$ having zero
complex part.

\begin{claim} The linear transformation corresponding to this matrix is
unitary.
\end{claim}
\begin{proof}
To see that this matrix is unitary, it suffices to show that the $\ell_2$ weight
from entries with index $j$ is preserved in the entries with index $j+1\pmod
{n+1}$ after application of the corresponding transformation. Let any vector in 
$\bbC^{4(n+1)}$ be given; we decompose its entries into real and complex part, 
$u(j,k,\ell)+iv(j,k,\ell)$. For $j\neq 0$, suppose that the two nonzero entries 
in row $(j,k,\ell)$ are columns $(j-1,k',\ell')$ and $(j-1,k'',\ell'')$, where
the former has weight with zero complex part, and the latter has zero real part.
Then, the output entry $(j,k,\ell)$ is
\[
\frac{1}{\sqrt{2}}(u(j-1,k',\ell')-v(j-1,k'',\ell''))
+\frac{i}{\sqrt{2}}(u(j-1,k'',\ell'')+v(j-1,k',\ell'))
\]
so its contribution to the $\ell_2$ weight is
\[
\frac{1}{2}((u(j-1,k',\ell')-v(j-1,k'',\ell''))^2+(u(j-1,k'',\ell'')+v(j-1,k',\ell'))^2)
\]
where, in the other row with columns $(j-1,k',\ell')$ and $(j-1,k'',\ell'')$
in its support, the contribution to the $\ell_2$ weight is
\[
\frac{1}{2}((u(j-1,k'',\ell'')-v(j-1,k',\ell'))^2+(u(j-1,k',\ell')+v(j-1,k'',\ell''))^2)
\]
and therefore, summing over these rows gives that the entries with index $j-1$ 
yield $\ell_2$ weight
\[
\sum_{k,\ell}(u(j-1,k,\ell)^2+v(j-1,k,\ell)^2)
\]
in entries with index $j$ (again, for $j\neq 0$) of the output. We also 
similarly find, for $j=0$, that the output entry $(0,k,\ell)$ is
\[
(\sqrt{1-\eta}u(n,k,\ell)-\sqrt{\eta} v(n,k\oplus 1,\ell))+
i(\sqrt{\eta} u(n,k\oplus 1,\ell)+\sqrt{1-\eta}v(n,k,\ell))
\]
so its contribution to the $\ell_2$ weight is
\begin{align*}
&(1-\eta)u(n,k,\ell)^2-2\sqrt{\eta(1-\eta)}u(n,k,\ell)v(n,k\oplus 1,\ell)
+\eta v(n,k\oplus 1,\ell)^2\\
&\qquad\qquad+\eta u(n,k\oplus 1,\ell)^2+
2\sqrt{\eta(1-\eta)}u(n,k\oplus 1,\ell)v(n,k,\ell)
+(1-\eta) v(n,k,\ell)^2
\end{align*}
where, summing over $(0,0,\ell)$ and $(0,1,\ell)$ gives
\[
u(n,0,\ell)^2+v(n,0,\ell)^2+u(n,1,\ell)^2+v(n,1,\ell)^2
\]
and hence, summing over all $(j,k,\ell)$ in the output, we observe that the
$\ell_2$ norm is indeed preserved, so the linear transformation is unitary.
\end{proof}

\paragraph{Choice of measurement and start state: }We let the Quantum 
Generator's measurement operator be as follows: for $j\notin\{0\}\cup S$, the 
basis states of the form $(j,k,b)$ are in the basis of the subspace
corresponding to the outcome $b$; for $j\in S-\{\min(S)\}$, the basis states
satisfying $(j,\ell\oplus b,\ell)$ are in the basis corresponding to the 
outcome $b$; and otherwise, the basis state $(j,b,\ell)$ is in the basis of the 
subspace corresponding to the outcome $b$. We take our start state to be the
basis state $(0,0,0)$. By the previous claim, this is a $4(n+1)$-state Quantum 
Generator, as promised.

\paragraph{Correctness: }
We are now in a position to verify that the $(n+1)$-symbol output distribution
of the constructed Quantum Generator is the distribution of noisy random
labeled examples of $f_S$.

\begin{claim}\label{basis-states}
Each $\ket{\psi_t}$ is of the form $\rho e_{(j,k,\ell)}$ where
$e_{(j,k,\ell)}$ is a vector corresponding to the basis state $(j,k,\ell)$ and
$\rho\in\bbC$ satisfies $|\rho|=1$.
\end{claim}
\begin{proof}
This claim is easy to verify by induction on $t$: assuming it is true of
$\ket{\psi_t}$, we see by inspection that the two entries in the support of
column $(j,k,\ell)$ in our matrix correspond to different measurement outcomes,
so the projection selects exactly one of them for $\ket{\psi_{t+1}}$.
\end{proof}

\begin{claim}\label{parities}
For $t=j\pmod{n+1}$, such that $j\geq\min(S)$, the Quantum Generator is in a 
basis state $(j,k,\ell)$ where $k=\bigoplus_{t':t'>t-j,t'\pmod{n+1}\in S}
x_{t'}$.
\end{claim}
\begin{proof}
Note first that if $t=\min(S)\pmod{n+1}$, then $\bigoplus_{t':t'>t-j,t'\pmod{n+1
}\in S}x_{t'}=x_t$. Thus, since by Claim~\ref{basis-states} $\ket{\psi_{t-1}}$ 
was a basis state, by construction we obtain $x_t=1$ if $\ket{\psi_t}$ is
supported by $(t\bmod{n+1},1,\ell)$ and $x_t=0$ when $\ket{\psi_t}$ is supported
by $(t\bmod{n+1},0,\ell)$. Suppose then for induction that this holds up to
$t\pmod{n+1}-1>\min(S)$. Then, if $t\pmod{n+1}\in S$, we see that by
construction, if $k=b_0$ and $x_t=b$, $\ket{\psi_t}$ is supported by the basis 
state $(t\bmod{n+1},b_0\oplus b,b_0)$, as needed. Otherwise,
$\ket{\psi_t}$ is supported by the basis state $(t\bmod{n+1},b_0,b)$ so in any 
case, the claim holds.
\end{proof}

We now observe that for $t\neq 0\pmod{n+1}$, by Claim~\ref{basis-states} and
further inspection, $\ket{\psi_t}$ is supported by a basis state $(t\bmod{n+1},
k,\ell)$ in the support of the measurement outcome 0 with probability $1/2$,
and is similarly in the support of the measurement outcome 1 with probability
$1/2$, so each such $x_t$ is uniformly distributed on $\{0,1\}$. Moreover, by
Claim~\ref{parities}, for $t=n\pmod{n+1}$, $\ket{\psi_t}$ is supported on a
basis state of the form $(n,b,\ell)$ for $b=\bigoplus_{t'\in\{t,t-1,\ldots,
t-n+1\}:t'\pmod{n+1}\in S}x_{t'}$. Thus, by construction, $\ket{\psi_{t+1}}$ is
supported by a basis state of the form $(0,b,\ell)$ with probability $1-\eta$ 
and of the form $(0,b\oplus 1,\ell)$ with probability $\eta$; since these
correspond to measurements $x_t=b$ with probability $1-\eta$ and $x_t=b\oplus 1$
with probability $\eta$, we see that every $(n+1)$ symbols of the output of this
Quantum Generator are distributed precisely according to the distribution of 
independent random labeled examples of $f_S$ with noise rate $\eta$, as desired.

\paragraph{Hardness of learning a parity distribution: }
Suppose that we could efficiently learn the output distribution of this Quantum
Generator. In particular, for any desired $\varepsilon$ we can therefore 
efficiently learn a circuit $E$ such that $KL(P_S||E)\leq\varepsilon
(1-H(\eta))$, where $H$ is the binary entropy function. For this circuit $E$,
observe that if $E(x,f_S(x))\leq E(x,\neg f_S(x))$, then it is easy to verify by
elementary calculus (the minimum is achieved at $E(x,f_S(x))=E(x,\neg f_S(x))$)
that $x$ contributes
\[
\frac{1}{2^n}\left(
\eta \log\frac{1}{E(x,\neg f_S(x))}+
(1-\eta)\log\frac{1}{E(x,f_S(x))}
\right)\geq \frac{1}{2^n}(1+\log\frac{1}{E(x)})
\]
to $KL(P_S||E)$. On the rest of the distribution, $E$ certainly encodes $P_S$
no better than the optimal encoding for $P_S$, so we find that if more than a 
$\varepsilon$ fraction of $x$ satisfy $E(x,f_S(x))\leq E(x,\neg f_S(x))$, then
\[
KL(P_S||E)> \varepsilon(1+n)+(1-\varepsilon)(H(\eta)+n)-(H(\eta)+n)
=\varepsilon(1-H(\eta))
\]
contradicting our assumption about the KL-divergence of $E$ from $P_S$.
Therefore we find that, for a uniformly chosen $x\in\{0,1\}^n$, the circuit $E'$
that outputs $b$ iff $E(x,b)>E(x,\neg b)$ correctly predicts $f_S(x)$ with
probability at least $1-\varepsilon$. This simple modification of $E$ can be
output efficiently, contradicting the assumed hardness of learning noisy
parities.
\end{document}